\documentclass[12pt]{article}
\usepackage{amsfonts,amssymb,amsmath,latexsym,epsfig,epic}       
\usepackage{ae,aecompl}
\usepackage{mathrsfs}
\usepackage{color}

\newcommand{\yes}{{\sc yes}}
\newcommand{\no}{{\sc no}}
\newcommand{\pa}{{\mathscr{P}}}

\newcommand{\cF}{{\mathcal F}}

\newcommand{\cQ}{{\mathcal Q}}

\newcommand{\cS}{{\mathcal S}}

\newcommand{\bQ}{{\bf Q }}
\newcommand{\bA}{{\bf A }}

\newcommand{\qp}{q^{p}}

\def\qed{\hfill\rule{2mm}{2mm}}

\begin{document}

\newtheorem{theorem}{Theorem}[section]
\newtheorem{lemma}{Lemma}[section]
\newtheorem{corollary}{Corollary}[section]
\newtheorem{claim}{Claim}[section]
\newtheorem{proposition}{Proposition}[section]
\newtheorem{definition}{Definition}[section]
\newtheorem{fact}{Fact}[section]
\newtheorem{example}{Example}[section]
\newcommand{\proof}{\em Proof: \em}

\title{Computing Majority with Triple Queries\\
} 
	
\author{Gianluca De Marco\footnotemark[1]
\and
Evangelos Kranakis\footnotemark[2]
\and
G\'abor Wiener\footnotemark[3]
}
\footnotetext[1]{
Dipartimento di Informatica e Applicazioni,
Universit\`a di Salerno,
84084 Fisciano (SA), Italy.
Email: \texttt{demarco@dia.unisa.it}
}
\footnotetext[2]{
School of Computer Science, Carleton University,
Ottawa, ON,  K1S 5B6, Canada. Supported in part by NSERC and MITACS grants.
Email: \texttt{kranakis@scs.carleton.ca}
}
\footnotetext[1]{
Department of Computer Science and Information Theory,
Budapest University of Technology and Economics, 
H-1521, Budapest, Hungary. Supported in part by the Hungarian
National Research Fund and by the National Office for Research and Technology (Grant Number OTKA 67651).
Email: \texttt{wiener@cs.bme.hu}
}


\maketitle

\begin{abstract}
Consider a bin containing $n$ balls colored with two colors.
In a $k$-query, $k$ balls are selected by a questioner and the
oracle's reply is related (depending on the computation model 
being considered) to the
distribution of colors of the balls in this 
$k$-tuple;
however, the oracle never
reveals the colors of the individual balls. 
Following a number of queries the questioner is said
to determine the majority color if it
can output a ball of the majority color if it exists,
and can prove that there is no majority if it does not exist.
We investigate two computation models (depending on the type of replies
being allowed).
We give algorithms to compute the minimum number of $3$-queries which are
needed so that the questioner can determine the majority color
and provide tight and almost tight upper and lower bounds 
on the number of queries needed in each case.

\vspace{0.4cm}
\end{abstract}

\section{Introduction}
We are given a bin containing $n$ balls colored with
two colors, e.g., 
EK red and blue.
At any stage, we can choose any $k$ of the balls and ask the question 
``do these balls have the same color?''.
The aim is to produce a ball of the {\em majority color}
(meaning that the number of balls with that color is strictly greater than that of the
other color), or to state that there is no majority 
(i.e. there is the same number of red and blue balls)
using as few questions as possible. We are considering the worst case problem, that is our aim is to find the number of questions the fastest algorithm uses in the worst case. 

\subsection{Model of computation}

In computing the majority
there are two participants: a {\em questioner} 
(denoted by $\bQ$) and an {\em oracle} (or {\em adversary})
denoted by $\bA$
\begin{figure}[!ht]
\begin{center}
      \includegraphics[height=2.5cm]{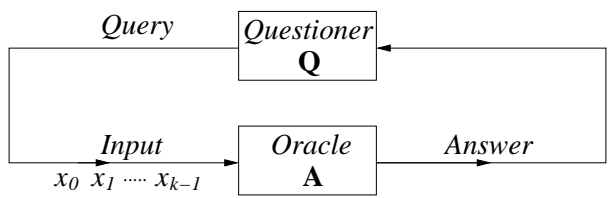}
    \caption{A questioner ($\bQ$), an oracle (or adversary) 
    ($\bA$) and a $k$-query input.}
        \label{fig:oracle}
\end{center}
  \end{figure}
(see Figure~\ref{fig:oracle}).
The questioner asks questions on the status of the color of the
balls by providing the oracle with $k$ balls 
and the oracle provides a reply which, depending on
the type of the oracle, is related to the
distribution of colors of the balls concerned.

Although the answer of the oracle depends on 
the $k$-tuple provided by the questioner, it is 
otherwise {\em permutation-independent}, i.e., the answer is
independent of the order of the balls.  Moreover, 
the
balls are endowed with distinct identities (say, integers $1,2,\ldots, n$).
Notice that in the course of querying the questioner is never allowed
to see the color of any ball but merely relies on the answer
provided by the oracle. 

\paragraph{Majority Problem.} 
Given $n$ balls of two colors, for a solution to the majority problem
we are interested to determine  
\begin{enumerate}
\item
whether or not there is a color occurring in the majority 
(i.e., more than half) of the balls, and 
\item
if indeed there is majority,
to output a ball having the majority color.
\end{enumerate}
In addition, we would like to minimize the number
of queries required. 

\paragraph{Query Models.}
Next we define two query models which will be considered in the sequel.
In each model the input to the oracle is a set 
$\{ x_0, x_1 , \ldots , x_{k-1} \}$ of $k$ balls.
\begin{enumerate}
\item
{\bf Y/N Model:}
The answer 
is either \yes~ or \no.
\yes~ means that all balls  have the same
color, \no~ means that the balls in the $k$-tuple do not
have the same color. 
\item
{\bf Pairing Model:}
The answer to a query 
is either \yes~ or \no. \yes~ means that all balls  have the same
color. \no~ means that not all the balls  have the same color 
and to show this, (any) two balls of different color are also provided.
\end{enumerate} 


\subsection{Notation}

Let $\qp_k (n), q_k (n)$ denote the minimum 
number of queries which are sufficient
to solve the majority problem for $n$ balls colored with two
colors using $k$-queries in the {\em
Y/N} and {\em Pairing} models, respectively. It is obvious that 
$\qp_k (n) \leq q_k (n)$ (assuming that these numbers exist). 
 
\subsection{Related work}

Our study 
is a natural generalization of the well-known {\em majority problem} \cite{A1} where at any stage 
{\em two} balls $a$ and $b$ are chosen and the question is ``do $a$ and $b$ have the same color?''. It is obvious that $q_2(n) = \qp_2 (n)$ and 
Saks and Werman \cite{SW} showed that $q_2(n) = n-\nu (n)$, 
where $\nu (n)$ denotes the number of 1's in the binary representation of
$n$.
 Later Alonso, Reingold, and Schott \cite{ARS2} also gave the solution for
the average case.  
Aigner \cite{A2,A3} introduced several variants and generalizations of the majority problem.
In particular, in the $(n,k)$-majority game one has to show a $k$-majority ball $z$ 
(that is, there are at least $k$ balls colored like $z$), or
declare there is no $k$-majority.

Other variations of the majority problem include the case where more than two colors are available. 
Fisher and Salzberg \cite{FS} studied the majority problem when the number of colors is any integer
up to $n$. In this case the majority color is the color such that there are at least
$n/2+1$ balls of that color. They solved the problem by showing that
$\lceil 3n/2 \rceil -2$ comparisons are necessary and sufficient.
Another natural generalization is the plurality problem where the goal is to identify a
ball of the dominant color (i.e., the one occurring more often than any other). In \cite{ADM} linear
upper and lower bounds were given for 3 colors. 
The currently best result for any number of colors 
was
given in \cite{CGMY}. The authors of \cite{CGMY} studied also the oblivious versions both for the majority
and the plurality problem. In oblivious strategies the questioner has to ask all questions in advance 
before getting any answer from the oracle. Finally, bounds for randomized algorithms can be 
found in \cite{DJKKS,KST}.

\subsection{Outline and results of the paper}


Section~\ref{existsec} 
discusses the problem of existence of a solution
to the majority problem in the
models proposed.
In Section~\ref{yn} we are considering the Y/N model for $k=3$ and give lower and upper bounds for $q_3 (n)$ whose difference is between 1 and 3, depending on the residue of $n$ modulo 4. 
Section~\ref{pairing} investigates the Pairing model. Here we give a general lower bound for $\qp_k(n)$ and compute the
precise value of $\qp_3 (n)$, namely
$$
\qp_3 (n) = \left\{
\begin{array}{ll}
n/2 +1 & \mbox{if $n$ is even, and}\\
\lfloor n/2 \rfloor & \mbox{if $n$ is odd.} 
\end{array}
\right.
$$


\section{Existence of Solutions}
\label{existsec}

Before trying to compute $\qp_k(n)$ and $q_k(n)$ we 
should discuss whether these numbers
exist at all, since it may happen that asking 
all possible queries is not enough for \bQ to solve the problem. 
It is clear that $n$ should be at least $k$ and if $n=k$, then 
the only possible query is enough in both games for $k=2$ and 
in the Y/N model for $k=3$, but not in the other cases. 
For $k\geq 3$ we prove that $q_k(n)$ and $\qp_k(n)$ exist if and only if
$n\geq 2k-2$ and $n\geq 2k-3$, respectively. 
\begin{theorem} \label{szeptetel}
Let $k\geq 3$. 
\begin{enumerate}
\item
The number $q_k(n)$ exists if and only if $n\geq 2k-2$.
\item
The number $\qp_k(n)$ exists if and only if $n\geq 2k-3$.
\end{enumerate}    
\end{theorem}
The proof can be found in the Appendix.

\section{Y/N Model}
\label{yn}

It might be worth 
observing
that in the Y/N model while a \no\ answer on a pair of balls 
tells us that these balls have different
colors and consequently they can be discarded by the questioner,  
when the number of balls is greater than two then
a \no\ answer only 
tells us
that there are (at least) two balls of different color. 
Therefore, if on the one hand it seems more advantageous comparing more balls at a time, 
on the other hand it is more challenging for the questioner to exploit a less informative \no\ answer. 

First we give upper bounds and next we discuss lower bounds
on the number of queries. We conclude with some examples. 

\subsection{The upper bound}
In this section we give an algorithm for solving the majority problem. 
We start with a straightforward extension of a result of Saks and Werman.
 
\begin{lemma}
\label{oddlm}
For all $n$ odd we have $q_k (n) \leq q_k (n-1) $.
\end{lemma}
The proof is basically the same as the proof of Saks and Werman for $k=2$; we include it in the Appendix for completeness.

\noindent
Let us now describe our algorithm.  

\paragraph{\bf Algorithm {\sc Majority}$_3$.}
Let $m = \lfloor n/4 \rfloor$. Consider an arbitrary partition of the $n$ balls in
$m$ groups $G_1,G_2, \ldots, G_m$ of size $4$ each 
and let $R$ be the group of the $r \equiv n \mod 4$ remaining balls.
For $i = 1,2,\ldots, m$, we make all ${4 \choose 3}=4$ possible triple queries 
on the four balls of $G_i$ till we get a \yes\ answer. There are two cases:
\begin{itemize}
\item
(a) for some $j$ we get a \yes\ answer on 
at least one of the $4$ triple queries involving balls of $G_j$;
\item
(b) we get always \no\ answers on all the $4m$ comparisons.
\end{itemize}

In case (a), we can discard all balls contained in $G_1,\ldots, G_{j-1}$ as they have no effect for
determining the majority. Let $a,b$, and $c$ the three balls in $G_j$ that have the same color.
Notice that we also know whether or not the color of the fourth ball $d$ in $G_j$ 
is the same as the color of the other three balls.
The number of queries required up to this point is
$4j$. From now on, all the remaining $n - 4j$ balls are compared, one at a time, with two balls 
of identical color, e.g. $a$ and $b$. 
It is clear that this way we can count the number of balls colored like $a$ and $b$ and the number of balls
colored like a different ball $e$, if it exists. This will allow us to
solve the problem by using a total number of at most $4j + (n - 4j) = n$ queries.

In case (b), we have found that there is an equal number of red and blue balls among the
$4m$ balls in $G_1 \cup G_2 \cup \cdots \cup G_m$.
Therefore, in this case, the majority is determined by the remaining $r$ balls in set $R$. 
Hence, we have to determine, for every $r\in\{0,1,2,3\}$, the number $q_3(4m+r)$.
We  limit our analysis to $q_3(4m)$ and $q_3(4m+2)$ as, in view of Lemma~\ref{oddlm},
we have $q_3(4m+1) \leq q_3(4m)$ and $q_3(4m+3) \leq q_3(4m+2)$. 
If $r=0$, the algorithm can already state that there is no majority.
If $r=2$, the problem reduces to ascertain whether these two remaining balls have the same color 
(in which case any of them is in the majority)
or not (which implies that there is no majority). In order to do so, we
compare the 2 remaining balls (call them $x$ and $y$) with three arbitrary balls $a,b,c$
(one at a time) from $G_1$. 
Namely, we perform the following queries: $\{x,y,a\}$, $\{x,y,b\}$ and $\{x,y,c\}$. 
If $x$ and $y$ are identically colored, we obtain a \yes\ (recall that, since we are in case (b), 
there must be two different balls among $a,b$, and $c$); otherwise we always obtain \no\ answers. 
Therefore, in both cases we can solve the ptoblem using at most 3 additional queries. 

In conclusion, we have proved the following.
\begin{lemma}\label{ub}
Let $n=4m+r$ for some $r\in \{0,1,2,3\}$ and $m\geq 1$. 
 \begin{itemize}
 \item[] $q_3(4m) \leq 4m = n$;
 \item[] $q_3(4m+1) \leq 4m  = n-1$;
 \item[] $q_3(4m+2) \leq 4m+3  = n+1$;
 \item[] $q_3(4m+3) \leq 4m+3  = n$.
\end{itemize}
\qed
\end{lemma}

The next lemma will be useful in the proof of the lower bound.
\begin{lemma}\label{inumber}
With at most one additional comparison, algorithm {\sc Majority}$_3$ is able to output a number $i$ such that 
there are $i$ balls of one color and $n-i$ of the other color.
\end{lemma}
The proof of Lemma~\ref{inumber} can be found in the Appendix.

\subsection{The lower bound}

In the sequel, a {\em coloring} is a partition $(R,B)$
of the set of balls into two sets $R$ and $B$, where $R$ is the set of red balls and $B$ the set of blue balls. 

In this section we will give a lower bound on the number of queries 
needed to determine the majority. 
Our aim will be to construct a worst case input coloring for any unknown correct algorithm 
that solves the majority problem.
We use the usual adversary lower bound technique. 


We say that an adversary's strategy {\em admits a coloring} if such a coloring is consistent
with all the answers provided by the adversary. 
As long as the strategy devised by the adversary admits alternative possible colorings that 
are consistent with at least two different possible solutions for the majority problem, the algorithm cannot 
correctly give its output. 
The goal of the adversary is to maximize, with its strategy of answers, the number of rounds
until the algorithm can correctly give its output.

We will first consider the case of an even number $n$ of balls.
Given a sequence of queries $\cS = (Q_1, \ldots , Q_t)$ and a positive integer $i\leq t$, 
let us define the following property:

$\pa(i):$ for all $X \subseteq [n]$, with $|X| = m$, 
there exists $j \leq i$ such that one of the following conditions hold:

$(a)$ $Q_j \subseteq X$;

$(b)$ $Q_j \cap X = \emptyset$.

{\bf The adversary's strategy.} The strategy followed by the adversary is defined as follows. 
To every query $Q_j$, the adversary replies \no\, as far as $\pa(j)$ does not hold. 
Let $i$ be the first index for which $\pa(i)$ holds, if it exists.
Then there exists a set $X \subseteq [n]$, with $|X| = m$, such that 
$Q_j \not\subseteq X$ and $Q_j \not\subseteq \bar{X} = [n]-X$ for $j=1,2,\ldots, i-1$ (this is because
$i$ is the smallest index such that $\pa(i)$ holds), and
either $Q_i \subseteq X$ or $Q_i \subseteq \bar{X}$ (these are conditions $(a)$ and $(b)$ on $\pa(i)$). 
From this point forth, the adversary replies
\yes\ to $Q_i$ and all subsequent queries $\{a,b,c\}$ if all three elements $a,b,c$ belong 
to the same set $X$ or $\bar{X}$, and replies \no\ to all the other queries. 



\begin{lemma}\label{even}
Let $n=2m$ for some $m>1$. Any algorithm that solves the majority problem on  $n$ balls uses at least $n-1$ queries.
\end{lemma}
\begin{proof}
 Assume that $Q_1,\ldots, Q_t$ is the sequence of queries that the algorithm produces before giving its output. 
We distinguish two main cases: 
either $\pa(t)$ does not hold or $\pa(t)$ holds.

{\bf Case 1:} {\em $\pa(t)$ does not hold}. 
In this case there exists $X \subseteq [n]$, with $|X| = m$, such that for all queries
$Q_1,\ldots, Q_t$ neither $(a)$, nor $(b)$ holds. According to its strategy, the adversary replies always \no\
in this case. 

If the algorithm states that there is a majority, the algorithm actually 
has to show a ball in $X$ or in $[n] \setminus X$. The adversary's strategy admits the following coloring: 
color red all the balls in $X$ and blue all the remaining balls.
Indeed, such a coloring is consistent with all the \no\ answers provided by the adversary, as
since property $\pa(t)$ does not hold, there is no query entirely contained in $X$ or in $\bar{X}$. This means that there is no majority, which contradicts the algorithm. 

Let us assume the algorithm states that there is no majority. 
Suppose that there exists a ball $x$ in $X$ such that 
$Q_i \not\subseteq \bar{X} \cup \{x\}$ for every $i=1,2,\ldots,t$.
In this case, the adversary's strategy allows us to produce the following coloring: 
color red all the balls in 
$\bar{X} \cup \{x\}$ and blue the others. Indeed, it is easy to observe that, since
$Q_i \not\subseteq \bar{X} \cup \{x\}$ for every $i=1,2,\ldots,t$, such a coloring is consistent with 
the \no\ answers provided by the adversary to each $Q_i$ for $i=1,2,\ldots,t$. This implies that there
is a majority: once again the claim of the algorithm is contradicted. The case in which there exists 
a ball $x$ in $\bar{X}$ such that $Q_i \not\subseteq X \cup \{x\}$ for every $i=1,2,\ldots,t$, is similar.
It remains to analyze the case when both the following conditions hold:

(1) for every $x \in X$ there exists a query $Q_i$ for some 
$i \in\{1,2,\ldots,t\}$, such that $Q_i \subseteq \bar{X} \cup \{x\}$;

(2) for every $x \in \bar{X}$ there exists a query $Q_i$ for some 
$i \in\{1,2,\ldots,t\}$, such that $Q_i \subseteq X \cup \{x\}$.

Condition (1) implies that the sequence of queries $Q_1, Q_2, \ldots, Q_t$ must contain a query
$\{x,x_1,x_2\}$ for every $x \in X$ and for some $x_1,x_2 \in \bar{X}$. Analogously, condition (2) implies
that there must be a query $\{x',x_1',x_2'\}$  for every $x' \in \bar{X}$ and for some $x_1',x_2' \in X$.
Since all these queries are clearly distinct, $t \geq 2m = n$.

{\bf Case 2:} {\em $\pa(t)$ holds}. In this case, the adversary replies \no\ to all the queries $Q_j$,
as far as $\pa(j)$ does not hold, and replies \yes\ to the first query $Q_i$ for which $\pa(i)$ holds. 
Since $\pa(i)$ holds, while for $j=1,2,\ldots,i-1$ $\pa(j)$ does not hold, there exists $X\subseteq [n]$, 
$|X|=m$, such
that either $Q_i\subseteq X$ or $Q_i\subseteq \bar{X}$, and 
\begin{equation}\label{X}
Q_j \not\subseteq X \mbox{ and } Q_j \not\subseteq \bar{X}\mbox{ for } j=1,2,\ldots,i-1.
\end{equation}

In the following, we assume without loss 
of generality that $Q_i\subseteq X$. 
Notice that this time the adversary's strategy admits any coloring consistent with all the \no\ answers 
given to $Q_1,Q_2, \ldots, Q_{i-1}$ and the \yes\ answer to $Q_i$.

If the algorithm concludes that there is a majority, it can be easily contradicted by observing
that the adversary's strategy admits a coloring that colors
red all the balls in $X$ and blue all the others.

Therefore, suppose the algorithm states that there is no majority. 
If there exists a ball $x \in X \setminus Q_i$ (resp. $y \in \bar{X}$) such that for every $j=1,2,\ldots,i$,
$Q_j \not\subseteq \bar{X} \cup \{x\}$ (resp. $Q_j \not\subseteq X \cup \{y\}$), 
the adversary's strategy admits the following coloring. Color 
red all the balls in $X \setminus \{x\}$ (resp. $\bar{X} \setminus \{y\}$) and blue the balls in $\bar{X} \cup \{x\}$
(resp. $X \cup \{y\}$), so to have a majority in the latter set. This contradicts the algorithm's claim. 

It remains to analyze only the case when both the following conditions hold:

(1*) for every $x \in X \setminus Q_i$ there is a query $Q_j$ for some $j\in \{1,2,\ldots,i-1\}$, such that
$Q_j \subseteq \bar{X} \cup \{x\}$;
 
(2*) for every $y \in \bar{X}$ there is a query $Q_j$ for some $j\in \{1,2,\ldots,i-1\}$, such 
that $Q_j \subseteq X \cup \{y\}$.

Let $\cF_1$ and $\cF_2$ be the set of queries necessary for (1*) and (2*) to be satisfied, respectively. 
Condition (1*) implies that the sequence of queries $Q_1, Q_2, \ldots, Q_{i-1}$ must contain at least a query 
$\{x,y,z\}$ for every $x \in X \setminus Q_i$ and for some $y,z \in \bar{X}$. Condition (2*) implies that 
there must be also at least a query $\{x',y',z'\}$  for every $x' \in \bar{X}$ and for some $y',z' \in X$.
Therefore,
\begin{equation}\label{first}
|\cF_1| + |\cF_2| \geq 2m - 3. 
\end{equation}
In order to estimate the total number of queries, we need also to consider the set $\cF_3$ of  
queries that involve balls in $Q_i= \{a,b,c\}$ and are not considered in $\cF_1$ and $\cF_2$. The query $Q_i$ clearly belongs to $\cF_3$, so  
\begin{equation}\label{second}
|\cF_3| \geq 1. 
\end{equation}
Now we have to distinguish two cases.

{\bf Case (i)}: there is no query $Q'=\{x,y,z\}$, such that $x\in \bar{X}$ and $y,z \in Q_i$.
In this case, $\cF_3$ must contain at least another query $Q$ that contains some element from $\{a,b,c\}$ and two elements from $\bar{X}$, otherwise the adversary's strategy would admit a coloring
where all the balls in $\bar{X} \cup \{a,b,c\}$ are red, which would imply the existence 
of a majority in this set. Hence, $\cF_3\geq 2$ and by (\ref{first}) the total number of queries is
$|\cF_1| + |\cF_2| + |\cF_3| \geq 2m-1$.

{\bf Case (ii)}: there exists a query $Q'=\{x,y,z\}$ such that $x\in \bar{X}$ and $y,z \in Q_i$.
Let $s$ be the number of queries involving one ball from $Q_i$ and two balls from $\bar{X}$. We may assume that
$s < 2$, otherwise $|\cF_3| \geq s + 1 \geq 2$ and the proof would be complete.
If $s < 2$, we must have that

(*) there exist two balls $a$ and $b$ in $Q_i$ such that there is no query involving 
$a$ or $b$ with two balls from $\bar{X}$.  

\noindent
Let us consider $x,y,z \in Q'$ and recall that $x\in \bar{X}$ and $y,z \in Q_i$. By (*) we have that

(**) there exists a ball 
$v\in \{y,z\}\subseteq Q_i$ such that there is no query $\{v,w_1,w_2\}$ with $w_1,w_2 \in \bar{X}$. 

\noindent
Moreover, since we are assuming that $|\cF_1| + |\cF_2| = 2m - 3$, condition (2*) implies that

(***) for every $x \in \bar{X}$ there is {\em exactly one} query $Q_x = \{x,y_x,z_x\}$ that includes 
$x$ and some $y_x,z_x \in X$.

Let $Y = X \setminus \{v\} \cup \{x\}$. We will show that there is no query among $Q_1,Q_2,\ldots, Q_i$
that is entirely contained in $Y$ or $\bar{Y}$, which contradicts the hypothesis that $\pa(t)$ holds with $i$ being the smallest index such that 
for all $X \subseteq [n]$, with $|X| = m$, 
there exists $j \leq i$ such that either $Q_j \subseteq X$ or $Q_j \subseteq \bar{X}$.

Since now $x\in Y$ and $v\in \bar{Y}$, we can easily verify that $Q_i \not\subseteq Y$ and $Q_i \not\subseteq \bar{Y}$.
It remains to show that no query among $Q_1,Q_2,\ldots, Q_{i-1}$ is entirely contained in $Y$ or $\bar{Y}$. 

As a consequence of (**), 
there is no query involving $v$ entirely contained in $\bar{Y}$. Moreover, by (***) 
$Q'$ must be the only query involving $x\in \bar{X}$ and two balls in $X$. Hence, 
there is no query involving $x$ entirely contained in $Y$. Let $\cQ$ be the set of remaining queries 
to analyze, i.e. the queries not involving $v$ nor $x$. By the definition of $Y$ and $\cQ$, it is straightforward
to observe that for all $Q\in \cQ$, $Q \subseteq Y$ (resp. $Q \subseteq \bar{Y}$) if and only if
$Q \subseteq X$ (resp. $Q \subseteq \bar{X}$). Therefore, since $i$ is the smallest index such that $\pa(i)$
holds, there is no query in $\cQ$ entirely included in $Y$ or in $\bar{Y}$. 
This is a contradiction and the proof is now complete.
\qed
\end{proof}

\begin{lemma}\label{odd}
For all $n$ odd we have $q_3 (n) \geq n-3$. 
\end{lemma}
The proof is based on Lemmas \ref{inumber} and \ref{even} and can be found in the Appendix.
 
Now we can state the main result in the following theorem.
\begin{theorem}
\label{mainthm}
Let $n=4m+r$ for some $r\in \{0,1,2,3\}$ and $m\geq 1$. We have:
\begin{itemize}
 \item[] $n-1 \leq q_3(4m)   \leq n$;
 \item[] $n-3 \leq q_3(4m+1) \leq n-1$;
 \item[] $n-1 \leq q_3(4m+2) \leq n+1$;
 \item[] $n-3 \leq q_3(4m+3) \leq n$.
\end{itemize}
\end{theorem}
\begin{proof}
The proof follows immediately by
combining Lemma~\ref{ub} with Lemmas~\ref{even} and \ref{odd}.
\qed
\end{proof}
Tightening the bounds of Theorem~\ref{mainthm} so as to compute
the exact value of $q_3 (n)$ does not seem to be
easy, but we provide the exact values of $q_3(n)$ for $n \leq 6$ in the appendix.

\section{Pairing Model}
\label{pairing}

In this section we focus on the Pairing model. 
First we prove a lower bound on $k$-tuples 
that does not depend on $k$ and then show that the bound is 
tight for $k = 3$. 


\subsection{A general lower bound}

First we prove a lemma.

\begin{lemma}
\label{lemma1}
Let $n\geq 2k-3, \;  k\geq 3$. Then 
$$
\qp_k (n) \geq \left\{
\begin{array}{ll}
n/2 +1 & \mbox{if $n$ is even, and}\\
\lfloor n/2 \rfloor & \mbox{if $n$ is odd.} 
\end{array}
\right.
$$
\end{lemma}
\begin{proof}
First we describe a strategy for the Adversary which guarantees 
that the number of queries the questioner should use is at least  
$\lfloor n/2 \rfloor $ if $n$ is odd. The strategy is quite simple, 
\bA always answers \no~ whenever it is possible, that is, if it is not 
known before the query that all $k$ balls asked have the same color 
(in which case the answer is \yes, of course, but such a question 
will never be asked). \bA also has to show two balls of different 
color; these can be any pair of balls that may have different color before the query. 

To see that \bQ has to use at least  $\lfloor n/2 \rfloor $ queries even 
against this simple strategy, we use graphs to 
describe the knowledge of \bQ (this is the same graph we use in the proof of Theorem \ref{szeptetel}). In ths Pairing model to every negative answer 
of \bA there corresponds a pair of balls having different colors. Let these 
pairs be the edges of a graph $G$, whose vertices are the balls. 

By the strategy of \bA there exists a coloring of the balls (with 
colors blue and red) such that there is no edge between two balls of 
the same color; therefore the graph $G$ is bipartite. Suppose now that 
\bQ can show a ball of majority color. This is possible 
if and only if there is a vertex $x$ of $G$ that always appears in the 
greater part in every bipartition of $G$. (Notice that now there is a 
majority color, since $n$ is odd.) Now it is easy to see that
$G$ has at least $\lfloor n/2 \rfloor$ edges (for completeness this appears in the Appendix as Lemma~\ref{belsolemma1}), thus the number of 
queries is also at least $\lfloor n/2 \rfloor$.

Now we slightly modify the above strategy to obtain the lower bound for $n$ even. \bQ can declare that there is no majority if and only if in every bipartition of $G$ the two parts have the same size, which is impossible if there is an isolated vertex in $G$. 

Suppose that \bA  answers the first $n/2 -1$ questions the same way as above. 
We consider two cases.

\paragraph{\bf Case 1.} The edges of $G$ after the $n/2 -1$ queries are not independent. Now the answer to the next query is also answered the same way as above. Since there is an isolated vertex in $G$ (we have $n/2$ edges that are not independent), 
\bQ cannot declare that there is no majority. Moreover, no vertex $x$ can always appear in the greater part in every bipartition of $G$, because there are only $n/2$ edges in $G$ (in the part where such an $x$ would appear, there would be at least $n/2+1$ non-isolated vertices).
This completes the proof in Case 1.

\paragraph{\bf Case 2.} The edges of $G$ after the $n/2 -1$ queries are independent. Now \bA has to be careful, since an edge between the two remaining isolated vertices would guarantee that there is no majority. So the edge is drawn somewhere else, which is possible, because $k\geq 3$ and there is no cycle in $G$ (that is, \bA may draw the edge between any two of the $k$ vertices of the query). Now the situation is the same as in Case 1. 
This completes the proof. \qed
\end{proof}

\subsection{Determining $\qp_3 (n)$}

Now we prove that the bound we have just proved in Lemma~\ref{lemma1} is tight for $k=3$. 
This is interesting because of several reasons. The strategy of \bA 
that gave the lower bound is quite simple and the lower bound is much 
smaller than the value of $q_2(n)$ (as well as $q_3 (n)$). However, 
the most surprising is 
that the function $\qp_3$ is not increasing, since $\qp_3(2n) = n+1$ 
and $\qp_3(2n+1)= n$. 
Actually, we have seen that  
$q_k$ is not strictly increasing and the same proof shows that the same is true for $\qp_k$:  $\qp_k(2n+1) \leq \qp_k(2n)$ and $q_k(2n+1) \leq q_k(2n)$.
Nevertheless, a complexity function that is not even increasing is 
quite unusual in search theory. 

The main theorem of this section is the following.
\begin{theorem}
\label{pairthm}
Let $n\geq 3$. Then 
$$
\qp_3 (n) = \left\{
\begin{array}{ll}
n/2 +1 & \mbox{if $n$ is even, and}\\
\lfloor n/2 \rfloor & \mbox{if $n$ is odd.} 
\end{array}
\right.
$$
\end{theorem}

\begin{proof}
We just have to show that the bounds in Lemma \ref{lemma1} are tight, 
that is, give an algorithm for the questioner \bQ that uses $n/2 +1$ 
queries if $n$ is even and $\lfloor n/2 \rfloor $ queries if $n$ is odd. 

The algorithm consists of two main parts of which the first is the 
formation of bins with balls and is similar to the one in \cite{SW}. 
Initially place all the $n$ balls in $n$ different bins and at each step 
do the following.
\begin{enumerate}
\item
Select any $3$ equal size bins and pick one representative
from each of them and query these $3$  balls.
\item
If the answer is \yes, then merge the bins into one new bin.
\item
If the answer is \no, then \bA 
specifies two balls of different colors. Remove the two corresponding bins.
\item
Iterate the process until there are no $3$ bins of
equal size.
\end{enumerate}

It is obvious that all remaining bins have size
a power of $3$. The process stops because we can no longer find
$3$ bins of equal size, that is, each bin size occurs at most twice.  
It is easy to check that to build a bin of size $3^c$ the
questioner \bQ needs $(3^{c}-1)/2$ queries, while to remove two bins of 
size $3^c$ each, the questioner \bQ needs exactly $3^c$ queries.  
Thus if the total number of removed balls is $2t$, then \bQ used $t$ 
queries  in order to get rid of them. 

This means that after the first main step of the algorithm, 
\bQ has $n-2t$ balls in bins whose sizes are powers of $3$, 
bin sizes occur at most twice and \bQ used $t$ queries to 
remove the $2t$ balls and $(3^{c}-1)/2$ queries to build 
a bin of size $3^c$. In other words, if the bin size $3^i$ 
occurs $a_i$ times ($i=0,1,2$) then the sequence 
$a_s,a_{s-1}, \ldots, a_0$ is the ternary expansion of 
$n-2t$ (where $s= \lceil \log_3 (n-2t) \rceil $) and \bQ 
used $t+\sum_{i=0}^s a_i (3^{i}-1)/2)$ queries altogether. 

It is obvious that if $a_s=0$ (which occurs iff $n-2t=0$), 
then there is no majority, if $a_s=1$, then the corresponding 
bin contains majority type balls, while if $a_s=2$, then \bQ 
cannot show a majority type ball, neither can prove that there 
is no majority at the moment, so we may suppose that $a_s =2$ 
(otherwise we are done, since $2t+ \sum_{i=0}^s a_i (3^{i}-1) \leq n$, 
so for the number of queries we have that
$t+ \sum_{i=0}^s a_i (3^{i}-1)/2) \leq n/2$). 

Now starts the second part of the algorithm, where \bQ tries to 
eliminate the remaining bins. If $s>0$, then \bQ chooses one ball 
from both bins of size $3^s$ and a third ball from one of these bins 
and queries this triple. Given the answer, \bQ learns if the two 
biggest bins have balls of the same color or not. If the answer 
is \yes, then all balls in these two bins are clearly in majority, 
if the answer is \no, then \bQ removes the two bins. Iterate this 
process if necessary (if the greatest bin size occurs twice), 
until it is possible, that is, while $s>0$. To remove two bins of 
size $3^c$ each, \bQ still needs exactly $3^c$ queries, therefore 
it is clear that at the end of the second part \bQ can either 
solve the problem (that is, find a majority ball or show that there 
is no majority) using at most $n/2$ queries or \bQ has just two bins 
of size $1$ and used exactly $(n-2)/2$ queries to reach this position 
(which is possible only if $n$ is even). Now to finish the algorithm 
and the whole proof let the two remaining balls be $a$ and $b$ and 
let $c$ and $c'$ be two balls of different color (such balls exist 
and are known to \bQ, since $n\geq 3$ and $n$ is even). Now \bQ 
queries the balls $a,b,c$. If the answer is \yes or the answer is 
\no and the two differently colored balls shown by \bA is $a$ and $b$, 
then \bQ can solve the problem using $n/2$ queries altogether. 
If the answer is \no and the two different balls shown by \bA 
are not $a$ and $b$, then \bQ queries the balls $a,b,c'$, thus 
solving the problem using $n/2 +1$ queries altogether.
This completes the proof of Theorem~\ref{pairthm}.
\qed
\end{proof}


\section{Conclusion}

In this paper we studied
the minimum number of triple queries needed to determine the majority color in a set of 
$n$ colored balls (colored with two colors) under 
two models a) Y/N, and b) Pairing, which depend on the type of queries
being allowed. In addition to 
tightening the bounds for the majority problem in the
Y/N model, several interesting questions remain open
for further investigation, including computing majority on 
a) $k$-tuple queries, for
some fixed $k$, b) bins with balls colored with more than two colors,
as well as c) other natural computation models. 


\newpage
\section*{Appendix}

\subsection*{Proof of Theorem~\ref{szeptetel}}

\begin{proof}
Consider $k$-tuple inputs.
We have to show that \bQ can solve the 
majority problem  by asking all possible queries 
if and only if $n\geq 2k-2$ in the Y/N model and if and only if $n\geq 2k-3$ in the Pairing model.
It is easy to see for both problems 
that if \bA gives a positive answer (i.e., declares that balls in a 
set $S$ of size $k$ have the same color), then \bQ can learn 
which balls have the same color as the balls in $S$, thus 
solving the problem. The questioner \bQ just has to ask every query
containing the balls of $S'$, where $S'$ is an arbitrary 
subset of $S$ having size $k-1$. 

Thus we may assume that \bA never gives a positive answer. 
Now the first part of the theorem is easy to prove. If $n \geq 2k-1$, then \bA must give a positive answer some time, since there are $k$ balls having the same color, while if $n = 2k-2$ and there are no $k$ balls of the same color, then there cannot be majority (both colors appear $k-1$ times). 
On the other hand, if $n \leq 2k-3$ then it is possible that there are exactly $k-1$ red balls and $n-k+1 \leq k-2$ blue balls, in which case the answer is always negative and the problem cannot be solved, since there is a majority but no ball can be named as one in majority. 


Since in the Pairing model 
\bQ is given more information, it is obvious that it can be solved if 
$n \geq 2k-2$. To show that in the case $n = 2k-3$, \bQ can also solve 
the problem, but in the cases $n \leq 2k-4$ cannot, we use graphs to 
describe the knowledge of \bQ. In this approach to every negative answer 
of \bA there corresponds a pair of balls having different colors. Let these 
pairs be the edges of a graph $G$, whose vertices are the balls. 

It is straightforward that there always exists a coloring of the balls 
(with colors blue and red) such that there is no edge between two balls 
of the same color; therefore the previously defined graph $G$ is bipartite. 
It is also obvious that \bQ can show a ball of majority color if and only 
if there is a vertex $x$ of $G$ that always appears in the greater part 
in every bipartition of $G$ and \bQ can declare that there is no majority 
if and only if in every bipartition of $G$ the two parts have the same size. 


Let us consider now the case $n \leq 2k-4$. Let $G$ be a graph on 
$n\leq 2k-4$ vertices with an edge set consisting of $\lfloor n/2 \rfloor- 
1$ independent edges. It is easy to see that any set of $k$ vertices 
spans at least one edge, so the answers of \bA are all negative and to 
every answer \bA can show a pair of balls that are neighbours in $G$. 

Assume now to the contrary that \bQ can solve the problem, that is in 
every bipartition of $G$ the two parts have the same size or there 
exists a vertex $x$ of $G$ that always appears in the greater part 
in every bipartition of $G$.

The first case is clearly impossible since there exists an isolated vertex in $G$. For the second case we prove a useful lemma.

\begin{lemma} \label{belsolemma1}
Let $G$ be a bipartite graph on $n$ vertices. If there exists a vertex $x$ of $G$ that always appears in the greater part in every bipartition of $G$, then $G$ has at least $\lfloor n/2 \rfloor$ edges.
\end{lemma}
\begin{proof} {\em (of Lemma \ref{belsolemma1})}. Let $C$ be the part of a bipartition of $G$
that contains $x$. It is easy to see that $C\setminus \{ x\}$ contains at least 
$\lfloor n/2 \rfloor$ non-isolated vertices, otherwise moving the isolated 
vertices of $C\setminus \{ x\}$  to the other part of the bipartition would 
leave $x$ in the smaller part, which is impossible. Notice that $x$ itself 
may be an isolated vertex. Since there is no edge between two vertices 
of $C$, there must be at least $\lfloor n/2 \rfloor$ edges in $G$. 
This completes the proof of Lemma \ref{belsolemma1}.
\qed
\end{proof}

By Lemma \ref{belsolemma1} the second case also leads to a contradiction, 
which shows that for $n \leq 2k-4$ the questioner \bQ cannot solve the problem.  


Now let us turn our attention to the only remaining case, $n=2k-3$. 
Since \bA must always answer \no, every set of $k$ vertices of $G$ spans at least one edge, in other words $\alpha (G)$, the maximum size of an independent vertex set of $G$ is at most $k-1$. It is well-known that $\alpha (G) + \tau(G) = n$, where $\tau (G)$ is  the minimum size of a vertex cover of $G$, thus $\tau(G) \geq  n - k +1$. $G$ is a bipartite graph, so by the theorem of K\H onig \cite{konig} $\tau (G) = \nu (G)$, where $\nu (G)$ is the maximum size of an independent edge set of $G$. Thus there exists an independent edge set $X$ of size $n-k+1 =  k-2 = (n-1)/2$ in $G$. Now the ball corresponding to the only vertex that is not covered by $X$ must be in majority, thus
finishing the proof of Theorem~\ref{szeptetel}.
\qed
\end{proof}

\subsection*{Proof of Lemma~\ref{oddlm}}

\begin{proof}

Let $S$ be the set of balls. Let us remove a
ball $x$ from $S$. By definition
we can solve the majority problem on $S \setminus \{ x \}$ using $q_k (n-1)$ queries.
Now there are two cases to consider. If 
$S \setminus \{ x \}$ has no majority, then the number
of blue and red balls in the set $S \setminus \{ x \}$ is the
same and therefore $x$ has the majority color in the set $S$. On the other hand if
$S \setminus \{ x \}$ has majority, then since $n-1$ is
even the number of balls of the majority
color is at least $(n-1)/2 + 2$ and therefore the majority color of $S$
must be the same as the majority color for
$S \setminus \{ x \}$, proving  $q_k (n) \leq q_k (n-1)$ for $n$ odd.
\qed
\end{proof}

\subsection*{Proof of Lemma~\ref{inumber}}

\begin{proof}
 In case (a) of the proof of Lemma~\ref{ub}, the statement of the lemma is straightforward even without the additional question. The same holds for case (b) if $r=0,1,2$. Let us consider now $r=3$. Let $R=\{x,y,z\}$ be the set of
remaining balls. The algorithm chooses any two among them 
and executes the same comparisons as for the case $r=2$ in order to ascertain whether these two balls are
identically colored or not. If they have different colors, the gap (difference between majority amd minority) is 1; if they have the same color,
one additional comparison among $x,y$, and $z$ will ascertain whether the gap is 1 or 3, finishing the proof.
\qed
\end{proof}

\subsection*{Proof of Lemma~\ref{odd}}

\begin{proof}
Let us consider a set $S$ of $m$ balls with $m$ even. Take 3 arbitrary balls $a,b$, and $c$ from $S$ and 
solve the majority problem on the remaining $m-3$ balls by using $q_3(m-3)$ queries. 
In view of Lemma~\ref{inumber}, using $q_3(m-3)+1$ questions we also know a number $i$, 
such that there are $i$ balls of the majority and $m -3 - i$ of the minority color. It is obvious that $i > (m-4) / 2$.

Let $x$ be the majority ball provided by the algorithm on the $m-3$ balls. Notice that $x$ exists as $m-3$ is odd.
We first show that with at most $4$ more queries we can solve the majority problem on the whole set $S$ of 
$m$ balls. First we perform the query $\{a,b,c\}$.

Assume first that we got a \yes\ answer on $\{a,b,c\}$. Now we perform query $\{a,b,x\}$. If we get a
\yes\ on $\{a,b,x\}$ we infer that $x$ is a majority ball for $S$. If $\{a,b,x\}$ gives a \no\ answer, then 
three cases are possible:
\begin{enumerate}
 \item[] if $i < \frac{m}{2}$ then either of $a$, $b$ or $c$ can be showed as majority ball;
 \item[] if $i = \frac{m}{2}$ then there is no majority ball;
 \item[] if $i > \frac{m}{2}$ then $x$ is a majority ball.
\end{enumerate}

If on the other hand we got a \no\ answer on $\{a,b,c\}$, we can solve the majority problem by using the following 
three additional queries: $\{a,b,x\}$, $\{a,c,x\}$ and $\{b,c,x\}$. 
If at least one of the three answers is
\yes\ then we can infer that $x$ is a majority ball. If we get \no\ on all three queries, then we infer 
that there is no majority among $a,b,c$ and $x$. Two cases are possible:
\begin{enumerate}
 \item[] if $i = \frac{m-4}{2} + 1$, then there is no majority;
 \item[] if $i > \frac{m-4}{2} + 1$, then $x$ is a majority ball.
\end{enumerate}
This proves that with a total of $(q_3(m-3)+1)+4 = q_3(m-3)+5$ queries we can solve the majority problem on the whole set $S$.
Therefore, recalling that on $m$ even, we know by Lemma~\ref{even} that $q_3(m) \geq m-1$, we have:
$q_3(m-3) + 5 \geq q_3(m) \geq m-1$,
which implies $q_3(m-3) \geq m-6$. The lemma follows by letting $n=m-3$. 
This completes the proof of Lemma~\ref{odd}. 
\qed
\end{proof}

\subsection*{Number of queries in the Y/N model, for $n\leq 6$}

In the sequel we compute the number of queries when the
number of balls is small ($n \leq 6$).

\paragraph{\bf Example: $q_3 (4) = 4$.}
There are ${4 \choose 3} = 4$ triples. By querying all
four of them we can decide whether or not there is 
majority: if all four answers are \no~ there is no
majority otherwise a \yes~ answer on a triple yields 
a ball of majority color.

\paragraph{\bf Example: $q_3 (5) = 4$.}
Observe that by Lemma~\ref{oddlm}, we have that $q_3 (5) \leq q_3 (4) =4$.
Assume to the contrary that $q_3 (5) \leq 3$ and let the
balls be numbered $1,2,3,4,5$. We show that the
majority problem cannot be solved with three queries. Consider the three
triples $Q_1, Q_2, Q_3$ that are being queried. 
If $x \not\in Q_1 \cup Q_2 \cup Q_3$, for some ball $x$,
then the three queries apply only to four of the five balls
$\{1,2,3,4,5\} \setminus \{ x\}$; therefore an ``adversary''
flipping the color of $x$ can change the outcome which shows that the majority problem
cannot be solved. Therefore it remains to consider the 
case $Q_1 \cup Q_2 \cup Q_3 = \{1,2,3,4,5\}$.
It is easy to see that there are two triples, say $Q_1, Q_2$,
such that $|Q_1 \cap Q_2| = 2$. W.l.o.g. assume
$Q_1 = \{ 1,2,3 \}$ and $Q_2 =\{ 2,3,4\}$. 
Clearly, $5 \in Q_3$ (otherwise we are in the previous case
and the majority problem cannot be solved).
There are (essentially) three cases to consider depending on the other two
balls in $Q_3$: 1) $2, 3 \in Q_3$, 2) $1, 3 \in Q_3$, and 3) $2, 4 \in Q_3$,
Again, a simple ``adversary argument''
shows that the majority problem
cannot be solved in any of these cases.

\paragraph{\bf Example: $q_3 (6) = 7$.}
Lemma~\ref{ub} tells us that $q_3(4m+2) \leq 4m+3 = n+1$. For $m=1$ 
it follows that $7$ queries are sufficient for $6$ balls.
What is more important is the following claim.

\paragraph{\bf Claim.} $6$ queries are not enough for $6$ balls.
Suppose the contrary. Since not all 10 pairs of complementary 
triples are asked, the $3$ red, $3$ blue situation is always possible 
if all the answers are \no. We show that if the triples are not in a 
special setting (described later), then the $4$ blue, $2$ red situation 
is also possible if all the answers are \no.

If the $4$-$2$ situation is not possible with just \no~ answers, then every 
$4$-tuple contains a triple we asked. In other words, the complements 
of our triples contain every possible pair of balls. So we would like 
to cover the edges of the complete graph on $6$ vertices with $6$ (or less) 
triangles. Since a triangle covers $2$ edges of the same end-vertex (i.e.
triangle $123$ covers $12$ and $13$) and every vertex has degree $5$ in $K_6$ 
(which is odd), this is possible if and only if three edges forming a 
perfect matching are double covered and all the others are simple covered 
($6$ triangles have $18$ edges and we have $15$ edges  (to see that $5$ triangle 
is not enough is even easier)). Let these edges be $12, 34, 56$.

Now we give an Adversary strategy: the first five questions are 
answered \no, the sixth is \no, if the above setting is not satisfied, 
otherwise it is \yes. In the first case both the $3$-$3$ and the $4$-$2$ 
situations are possible. We show the same for the second case. 
Now all triples contain either $12, 34$, or $56$ (since we have six 
triples and these edges are double covered), so we may suppose 
that the last triple was $456$ (that is, the last question was $123$). 
We know that $1, 2$, and $3$ are of the same color, but the other balls 
may all have the other color, since the only question (answered \no)~ 
that is not consistent with this is the question $456$ and it is easy 
to see that this cannot be a question, because in this case we have 
the triples $123$ and $456$ and all other four  triples have just one ball 
in common with one of them, so there are four double covered edges, 
which is impossible. It remains to show that the $4$-$2$ situation is 
also possible. We know that $123$ is a question, let the other 
question whose complement covers $56$ be $1234\setminus\{x\}$, where $x$ is $1,2$, 
or $3$. Let $y$ be an element from $1,2,3$ different from $x$. Then the 
answers are consistent with the situation that $1,2,3$, and $y$ are 
blue, the others are red. 
This finishes the proof.


\end{document}